\documentclass[10pt,conference]{IEEEtran}
\usepackage[english]{babel}
\usepackage{graphicx}
\usepackage{amsmath}
\usepackage{amsfonts}
\usepackage{amssymb}
\usepackage{amsfonts,eucal,bm}
\usepackage{a4wide}

\newtheorem{lemma}{Lemma}

\newtheorem{example}{Example}
\newtheorem{remark}{Remark}

\begin{document}

\title{On improving security of GPT cryptosystems}

\author{\authorblockN{Ernst M. Gabidulin}
\authorblockA{Department of Radio Engineering\\
Moscow Institute \\
of Physics and Technology \\
(State University)\\
141700 Dolgoprudny,  Russia\\
Email: gab@mail.mipt.ru } 
\and
\authorblockN{Haitham Rashwan }
\authorblockA{Department of Communications \\
InfoLab21, South Drive\\
Lancaster University\\
Lancaster UK LA1 4WA\\
 Email: h.rashwan@lancaster.ac.uk}
\and
\authorblockN{Bahram Honary}
\authorblockA{Department of Communications \\
InfoLab21, South Drive\\
Lancaster University\\
Lancaster UK LA1 4WA\\
Email: b.honary@lancaster.ac.uk }
}
\maketitle

\begin{abstract}
The public key cryptosystem based on rank error correcting codes (the GPT cryptosystem) was proposed in
1991. Use of rank codes in cryptographic applications is advantageous since it is practically
impossible to utilize combinatoric decoding. This enabled using public keys of a smaller size. Several
attacks against this system were published, including Gibson's attacks and more recently Overbeck's attacks. A
few modifications were proposed withstanding Gibson's attack but at least one of them was broken by
the stronger attacks by Overbeck. A tool to prevent Overbeck's attack is presented in  \cite{Gab2008}.  In
this paper, we apply this approach to other variants of the GPT cryptosystem.
\end{abstract}

\section{Introduction}
The first code-based public-key cryptosystem is introduced and investigated in \cite{McEliece}. The
system is based on Goppa codes in the Hamming metric. It is a strong cryptosystem but the
size of a public key is too large  for practical implementations to be efficient.

The public key cryptosystem based on \emph{rank} error correcting codes was proposed in \cite{GPT1991,
Gabidulin:1995/2} and is now called the GPT cryptosystem.

Rank codes are well structured. It makes easier creation of attacks. Subsequently in a series of works,
Gibson \cite{Gibson:1995, Gibson:1996} developed attacks that break the GPT system for public keys of
about $5$ Kbits which are efficient for practical values of parameters $n\leq 30$, where  $n$ is length
of rank codes with the field $\mathbb{F}_{2^{n}}$ as an alphabet.

Several variants  of the GPT PKC were introduced to withstand Gibson's attacks \cite{GabOuriv2000,
ColumnScrambler2003}. One proposal is use of a rectangular row scramble matrix instead of a square
matrix. This allows to work with subcodes of  rank codes having much more complicated structure.
Another proposal exploits a modification of Maximum Rank Distance (MRD) codes where the concept of a
\emph{column} scramble matrix was also introduced. A new class of rank codes, so called, reducible codes,
are also implemented  to modify the GPT cryptosystem \cite{ReducibleRankCodes2002,
HighWeightErrors2005}. All these variants withstand Gibson's attack.

Recently, R. Overbeck \cite{Overbeck:2005}, \cite{Overbeck:2008} proposed a new attack which is more effective than any of
Gibson's attacks. His method is based on the fact that a column scrambler is defined over the \emph{base field}. A
generalization and development of one Gibson's idea allows him to break many instances of the GPT
cryptosystem. It was found in \cite{Gab2008} that a cryptographer can define a proper column scrambler
over the \emph{extension field} \emph{without violation} of the standard mode of the PKC. It turns out that
Overbeck's attack fails in this case.

In this paper, we implement an idea of a proper choice of column scramblers over the extension field to
other variants of the GPT cryptosystem. This choice withstands Overbeck's attacks as well as Gibson's
attacks.

\section{The GPT cryptosystem}
\subsection{Rank codes}
Let $\mathbb{F}_q $ be a finite field of $q$ elements and let  $\mathbb{F}_{q^{N}} $ be an extension
field of degree $N$.

Let $\mathbf{x}=( x_{1},x_{2},\ldots ,x_{n}) $ be a vector with coordinates in $\mathbb{F}_{q^{N}} $.

The \textit{Rank} norm of $\mathbf{x%
}$ is denoted  $\mathrm{Rk}(\mathbf{x}\mid \mathbb{F}_q)$ and is defined as the \textit{maximal} number
of $x_{i}$, which are linearly independent over the \emph{base field} $\mathbb{F}_q $.

Similarly, for a matrix $\mathbf{M}$ with entries in $\mathbb{F}_{q^{N}}$ the \emph{column rank }is
defined as the \textit{maximal} number of columns, which are linearly independent over the base field
$\mathbb{F}_q $, and is denoted $\mathrm{Rk}(\mathbf{M}\mid\mathbb{F}_q)$.

The \textit{Rank} distance between $\mathbf{x}$ and $\mathbf{y}$ is defined as the rank norm of the
difference $\mathbf{x}-\mathbf{y}$,  i.e.\;
$d(\mathbf{x},\mathbf{y})=\mathrm{Rk}(\mathbf{x}-\mathbf{y}\mid \mathbb{F}_q)$.

The theory of optimal MRD (Maximal Rank Distance) codes is given in \cite{Gab1985}. A
\textit{generator} matrix $\mathbf{G}_k$ of a MRD code is defined by

\begin{equation}\label{GeneratorMatrix}
\mathbf{G}_k=\left[
\begin{array}{llll}
g_{1} & g_{2} & \cdots & g_{n} \\
g_{1}^{[1]} & g_{2}^{[1]} & \cdots & g_{n}^{[1]} \\
g_{1}^{[2]} & g_{2}^{[2]} & \cdots & g_{n}^{[2]} \\
\vdots & \vdots & \ddots & \vdots \\
g_{1}^{[k-1]} & g_{2}^{[k-1]} & \cdots & g_{n}^{[k-1]}
\end{array}
\right] ,
\end{equation}
where $g_{1},g_{2},\ldots ,g_{n}$ are any set of elements of  the extension field $\mathbb{F}_{q^{n}} $
which are linearly independent over the base field $\mathbb{F}_ q $.

The notation $g^{[ i] }:=g^{q^{i\mathrm{\mod} n}}$ means the $i$th Frobenius power of $g$.

A code with the generator matrix \eqref{GeneratorMatrix} is referred to as a $(n,k,d)$ code, where $n$ is the
code length, $k$ is the number of information symbols, $d$ is the code distance. For MRD codes, $d=n-k+1$.

Let $\mathbf{m}=(m_1,m_2,\dots,m_k)$ be an information vector of dimension $k$. The corresponding code
vector is the $n$-vector
\[
\mathbf{g}(\mathbf{m})=\mathbf{mG}_k.
\]

If $\mathbf{y}=\mathbf{g}(\mathbf{m})+\mathbf{e}$ and $\mathrm{Rk}(\mathbf{e})=s\le t=\frac{d-1}{2}$,
then the information vector $\mathbf{m}$ can be recovered uniquely from $\mathbf{y}$ by some decoding
algorithm.

There exist \emph{fast} decoding algorithms for MRD codes \cite{Gab1985, gab92}.   A decoding procedure
requires elements of the $(n-k)\times n$ parity check matrix $\mathbf{H}$ such that
$\mathbf{G}_k\mathbf{H}^{\top}=\mathbf{0}$. For decoding, the matrix $\mathbf{H}$ should be of the form
\begin{equation}\label{CheckMatrix}
\mathbf{H}=\left[
\begin{array}{llll}
h_{1} & h_{2} & \cdots & h_{n} \\
h_{1}^{[1]} & h_{2}^{[1]} & \cdots & h_{n}^{[1]} \\
h_{1}^{[2]} & h_{2}^{[2]} & \cdots & h_{n}^{[2]} \\
\vdots & \vdots & \ddots & \vdots \\
h_{1}^{[d-2]} & h_{2}^{[d-2]} & \cdots & h_{n}^{[d-2]}
\end{array}
\right] ,
\end{equation}
where elements $h_{1},h_{2},\ldots ,h_{n}$ are in the extension field $\mathbb{F}_{q^{n}} $ and are
linearly independent over the base field $\mathbb{F}_ q $.

\subsection{Description of the standard GPT cryptosystem}
The GPT cryptosystem is described as follows.
\subsubsection{Possible generator matrices
using as public keys} Denote by $\mathbf{G}_{\mathrm{pub}}$ the public key, which is a generator matrix
of a code.
\begin{enumerate}
  \item \begin{equation}\label{eq:Pub1}
  \mathbf{G}_{\mathrm{pub}}=\mathbf{S}\mathbf{G}_k\mathbf{P}.\end{equation}
  The main matrix $\mathbf{G}_k$ is given by Eq. \eqref{GeneratorMatrix}. It is used to correct rank
errors. Errors of rank not greater than $t=\left\lfloor\frac{n-k}{2}\right\rfloor$ can be corrected.

A square $k\times k$ matrix $\mathbf{S}$ over the extension field $\mathbb{F}_{q^{n}}$ is called the
row scrambling matrix. It is used to destroy any visible structure of the matrix $\mathbf{G}_k$ by
mixing its rows.

A matrix $\mathbf{P}=\begin{bmatrix}p_{ij}\end{bmatrix} $ is called the column scrambler. This matrix
is a non singular square matrix of order $n$. It is used to mix  columns of $\mathbf{G}_k$.

If $\mathbf{P}$ is a matrix over the \emph{base field} $\mathbb{F}_{q}$, then a matrix $\mathbf{G}_k\mathbf{P}$ has just the same structure as the matrix $\mathbf{G}_k$ with a different first row. Hence, from the point of view of breaking, matrices $\mathbf{G}_{\mathrm{pub}}=\mathbf{S}\mathbf{G}_k\mathbf{P}$ and $\mathbf{G}_{\mathrm{pub}}=\mathbf{S}\mathbf{G}_k$ are equivalent. A cryptographer may not use a matrix $\mathbf{P}$ at all.

On the other hand, if entries $p_{ij}$ are in the extension field $\mathbb{F}_{q^{n}}$, then a matrix $\mathbf{P}$ makes breaking much harder. We shall analyze this case.

  \item Another generator matrix is obtained by an extension of matrix $\mathbf{G}_k$:
  \begin{equation}\label{eq:Pub2}
  \mathbf{G}_{\mathrm{pub}}=\mathbf{S}\begin{bmatrix}\mathbf{X} & \mathbf{G}_k                               \end{bmatrix}
\mathbf{P}.
\end{equation}
A matrix $\mathbf{X}$ of size $k\times t_1$ is called a distortion source
matrix. This matrix is a part of the concatenation $\begin{bmatrix}\mathbf{X} & \mathbf{G}_k
\end{bmatrix}$. The column rank of $\mathbf{X}$ is $\mathrm{Rk}(\mathbf{X}\mid \mathbb{F}_q)=t_1$. The
number $t_1$ is a \textit{design} parameter. Another  \textit{design} parameter is the ordinary rank
which can take values from $1$ to $t_1$. The rank distance of a code generated by the matrix $\mathbf{G}_{\mathrm{pub}}$ is not less than the rank distance of a code generated by the matrix $\mathbf{S}\begin{bmatrix}\mathbf{O} & \mathbf{G}_k \end{bmatrix}\mathbf{P}$.

A matrix $\mathbf{P}$ is called the column scrambler. This matrix is a non-singular square matrix of
order $n+t_1$. It is used to mix and to corrupt columns of $\mathbf{G}_k$ by means  of the distortion
source matrix $\mathbf{X}$.

Note that in previous works, the matrix $\mathbf{P}$  \emph{has all its entries in the base field}
$\mathbb{F}_q$. Overbeck's attack against this PKC succeeded due to this fact. But the attack fails for
the proper choice of $\mathbf{P}$ over the extension field $\mathbb{F}_{q^{n}} $ \cite{Gab2008}.

  \item \begin{equation}\label{eq:Pub3}
  \mathbf{G}_{\mathrm{pub}}=\mathbf{S}\begin{bmatrix}\mathbf{X} & \mathbf{G}_k                               \end{bmatrix}
\mathbf{P}.
\end{equation}
Here a scrambling matrix $\mathbf{S}$ is a rectangular $(k-p)\times k$
matrix.
  \item \begin{equation}\label{eq:Pub4}
  \mathbf{G}_{\mathrm{pub}}=\mathbf S\Bigl([\mathbf O\ \mathbf
G_k]+[\mathbf X_1\ \mathbf X_2]\Bigr)\mathbf{P}.
\end{equation}
Here: the row scrambler $\mathbf{S}$ is a
square non-singular matrix of order $k$ with entries in $\mathbf \mathbb{F}_{q^n}$ chosen at random; $\mathbf O$
is the $k\times m$ matrix of $0$'s; $\mathbf X_1$ is some $k\times m$ matrix  --- the first distortion
matrix; $\mathbf X_2$ is a $k\times n$ matrix with $r(\mathbf X_2|\mathbf F_1)=t_1$ --- the second
distortion matrix; the column scrambler $\mathbf P$ is a non-singular matrix of order $n+m$ with
entries in $\mathbb F_{q}$.
\end{enumerate}
\subsubsection{Plaintext} For public keys \eqref{eq:Pub1}, \eqref{eq:Pub2} and \eqref{eq:Pub4}, a \textbf{plaintext }is any $k$-vector $\mathbf{m}=$ $( m_{1},m_{2},\ldots
,m_{k})$, $m_{s}\in \mathbb{F}_{q^{n}} ,\ s=1,2,\ldots ,k$. For the public key \eqref{eq:Pub3}, a
plaintext is a $(k-p)$-vector.

\subsubsection{Private keys} The \textbf{Private keys} are matrices $\mathbf{S,G_k,X,,X_1,X_2,P}$ separately
and (explicitly) a fast decoding algorithm of an MRD code. Note also, that the matrices
$\mathbf{X},\mathbf{X}_1,\mathbf{X}_2$ are not used to decrypt a ciphertext and can be deleted after
calculating the Public key.

\subsubsection{Encryption} Let $\mathbf{m}=( m_{1},m_{2},\ldots ,m_{k}) $ be a plaintext. The
corresponding ciphertext is given by
\begin{equation}\label{Ciphertext}
\mathbf{c}=\mathbf{mG}_{pub}+\mathbf{e}=\mathbf{mS}[\mathbf{X}|\mathbf{G}_k]\mathbf{P}+ \mathbf{e},
\end{equation}
where $\mathbf{e}$ is an artificial vector of errors of  rank $t_{2}$ or less, randomly chosen and
added by the sending party. The number $t_2$ is the third design parameter.

\subsubsection{Decryption} The legitimate receiver upon receiving $\mathbf{c}$  calculates
\[
\mathbf{c}^{\prime }=\mathbf{cP}^{-1}=\mathbf{mS}[\mathbf{X}|\mathbf{G}_k]+\mathbf{eP}^{-1}.
\]
Then  he extracts from $\mathbf{c^{\prime }}$ the plaintext $\mathbf{m}$ using decoding algorithms and
properties of public keys.

\section{The Overbeck attack - an idea}
In \cite{Overbeck:2005, Overbeck:2008}, a new attack is proposed on the GPT PKC described by means  of Eq.
\eqref{eq:Pub2}.

It is claimed, that  similar attacks can be proposed on all the variants of GPT PKC.

We can not describe the attack in detail but recall briefly an idea of this attack.

We need some notations.

For $x\in \mathbb{F}_{q^n}$, let $\boldsymbol{\sigma}:\ \mathbb{F}_{q^n}\rightarrow \mathbb{F}_{q^n},\ \sigma(x)=x^q$ be the Frobenius automorphism.

For the matrix $\mathbf{T}=(t_{ij})$ over $\mathbb{F}_{q^{n}}$, let
$\sigma(\mathbf{T})=(\sigma(t_{ij}))=(t_{ij}^q)$.

For any  integer $s$, let $\sigma^s(\mathbf{T})=\sigma(\sigma^{s-1}(\mathbf{T}))$.

It is clear that $\sigma^n=\sigma$. Thus the inverse exists $\sigma^{-1}=\sigma^{n-1}$.

The following simple properties of $\boldsymbol{\sigma}$ are useful:
\begin{itemize}
  \item $\sigma(a+b)=\sigma(a)+\sigma(b)$.
  \item $\sigma(ab)=\sigma(a)\sigma(b)$.
  \item In general, for matrices $\sigma(\mathbf{T})\neq \mathbf{T}$.
  \item If $\mathbf{P}$ is a matrix over the \emph{base field}  $\mathbb{F}_q$, then
  $\sigma(\mathbf{P})=\mathbf{P}$.
  \end{itemize}

\subsubsection{An idea of Overbeck's attack} To break a system, a cryptanalyst constructs from the
public key $ \mathbf{G}_{\mathrm{pub}}=\mathbf{S}\begin{bmatrix}\mathbf{X} & \mathbf{G}_k
\end{bmatrix} \mathbf{P}$ the \emph{extended } public key as follows:
  \begin{equation}\label{ExtendedKey1}
    \mathbf{G}_{ext,pub} =
    \left \| \begin{array}{r}
        \mathbf{G}_{pub}\; \\
        \sigma(\mathbf{G}_{pub}) \\
        \sigma^2(\mathbf{G}_{pub}) \\
        \ldots \\
        \sigma^u(\mathbf{G}_{pub})
    \end{array} \right \| .
\end{equation}

The property that $\sigma(\mathbf{P})=\mathbf{P}$, if $\mathbf{P}$ is a matrix over the \emph{base
field} $\mathbb{F}_q$,  is used in \eqref{ExtendedKey1}. Further transformations of Eq.
\eqref{ExtendedKey1} allows to obtain the first row of the check matrix $\mathbf{H}$ of the rank code
used. It is enough to break the cryptosystem.

If $\mathbf{P}$  is a matrix over the \emph{extension field} $\mathbb{F}_{q^n}$, then
$\sigma(\mathbf{P})\neq\mathbf{P}$.

\textbf{We have to stress that  Overbeck's attack fails in this case.}

Moreover Gibson's attacks use also in implicit form the condition $\sigma(\mathbf{P})=\mathbf{P}$ and
can not be implemented without it.

Our intention is to show that there exist column scramblers $\mathbf{P}$ in  the \emph{extension field}
$\mathbb{F}_{q^n}$ such that the GPT PKC works and is secure against all known attacks.
\section{Other attacks on the GPT PKC}
An important part of a decryption procedure is correcting rank errors using a fast decoding algorithm known to the legitimate party. An unauthorized party may want to correct rank errors by a general algorithm without any knowledge of the structure of a rank code. We consider algorithms described in \cite{OurivskiJohansson2002} and  in the recent paper \cite{Perret2008}.

The authors of  \cite{OurivskiJohansson2002}  proposed two  algorithms for decoding an arbitrary $(n, k)$ linear rank distance code over $\mathbf{F}_{q^N}$. These algorithms correct errors of rank $t=\left\lfloor\frac{n-k}{2}\right\rfloor$ in $O\left((Nt)^3q^{(t-1)(k+1)}\right)$
and
$O\left((k + t)^3t^3q^{(t-1)(N-t)}\right)$ operations in $\mathbb{F}_q$ respectively.

Consider as an example a case when we use a $(28,14)$ rank code with $N=n=28,k=14, q=2, d=15, t=7$. The size of the public key is equal to $Nnk=10976$ bits. To correct $7$-fold rank errors, Ourivski--Johansson's algorithms  \cite{OurivskiJohansson2002} require  $2^{113}$ and $2^{147}$ operations in $\mathbb{F}_2$. Thus these attacks are infeasible for practical implementations.

The algorithm of \cite{Perret2008} requires $O\left(\log(q)N^{3(N-t)}\right)$ operations. We have for the above example  $2^{302}$ operations. Thus this attack is also infeasible for practical implementations.

\section{The simple GPT PKC}
Consider the public key of Eq. \eqref{eq:Pub1}. No distortion matrix $\mathbf{X}$ is used. A ciphertext
has the form
\begin{equation}\label{SimplePKC}
\mathbf{c}=\mathbf{mS}\mathbf{G}_k\mathbf{P}+ \mathbf{e},
\end{equation}
where the rank $\mathrm{Rk}(\mathbf{e}\mid \mathbb{F}_q)=t_1$ of an artificial error $\mathbf{e}$ is less or equal to $t=\lfloor\frac{n-k}{2}\rfloor$.

Brute-force attacks are based on the exhaustive search of possible artificial errors $\mathbf{e}$. It depends on the number of error vectors. If artificial errors are all possible $n$-vectors of rank $t_1$, then the number of operations to search is $O\left(q^{nt_1}\right)$.

Attacks on the public key contemplate to find unknown factors (to a cryptanalyst) $\mathbf{S}$, $\mathbf{G}_k$ and $\mathbf{P}$, or, to find  matrices $\mathbf{\widetilde{S}}$, $\mathbf{\widetilde{G}}_k$ and $\mathbf{\widetilde{P}}$ such that $\mathbf{S}\mathbf{G}_k\mathbf{P}=\mathbf{\widetilde{S}}\mathbf{\widetilde{G}}_k\mathbf{\widetilde{P}}$  from the known public key matrix $\mathbf{S}\mathbf{G}_k\mathbf{P}$  .

Assume first that the column scrambler $\mathbf{P}$ is a matrix over the base field $\mathbb{F}_{q}$. The legitimate user knows the secret key $\mathbf{P}$ and $\mathbf{P}^{-1}$. His algorithm is as follows.
 \begin{enumerate}
   \item Get a ciphertext $\mathbf{c}=\mathbf{mS}\mathbf{G}_k\mathbf{P}+ \mathbf{e}$.
   \item Multiply to the right by $\mathbf{P}^{-1}$. Get an intermediate ciphertext
       \begin{equation}\label{eq:InterC}
       \mathbf{c}^{\prime}=\mathbf{c}\mathbf{P}^{-1}=\mathbf{mS}\mathbf{G}_k+ \mathbf{e}\mathbf{P}^{-1}.
       \end{equation}
        Note that $\mathrm{Rk}(\mathbf{e}\mathbf{P}^{-1}\mid \mathbb{F}_q)=\mathrm{Rk}(\mathbf{e}\mid \mathbb{F}_q)=t_1\le t=\lfloor\frac{n-k}{2}\rfloor$ since $\mathbf{P}^{-1}$ is in the \emph{base field } $\mathbb{F}_q$.
   \item Decode $\mathbf{c}^{\prime}$ using a fast decoding algorithm and get $\mathbf{mS}$.
   \item Get a plaintext $\mathbf{m}$ as $(\mathbf{mS})\mathbf{S}^{-1}$.
 \end{enumerate}

On the other hand, the cryptanalyst can get a successful
representation $\mathbf{G}_{\mathrm{pub}}=\mathbf{\widetilde{S}}\widetilde{\mathbf{G}}_k$ for the equivalent rank
code with the generator matrix $\widetilde{\mathbf{G}}_k$ from the public key $\mathbf{S}\mathbf{G}_k\mathbf{P}$. It can be  done by means of Gibson--Overbeck's attacks and therefore break the system.

The situation is quite different if $\mathbf{P}$ is a matrix over the extension field
$\mathbb{F}_{q^N}$. For the general matrix $\mathbf{P}$, it is unknown how to solve the following problems: to find the public key factors $\mathbf{S}$, $\mathbf{G}_k$ and $\mathbf{P}$, or, to find  matrices $\mathbf{\widetilde{S}}$, $\mathbf{\widetilde{G}}_k$ and $\mathbf{\widetilde{P}}$ such that $\mathbf{S}\mathbf{G}_k\mathbf{P}=\mathbf{\widetilde{S}}\mathbf{\widetilde{G}}_k\mathbf{\widetilde{P}}$ from the known public key matrix $\mathbf{S}\mathbf{G}_k\mathbf{P}$. Gibson--Overbeck's attacks are not applicable if a matrix $\mathbf{P}$ is chosen in the extension field $\mathbb{F}_{q^N}$.

We can assume  from now on that Gibson's and Overbeck's attacks can not be
implemented. But the cryptographer should  select a \emph{secret} column scrambler $\mathbf{P}$ in the \emph{extension field} $\mathbb{F}_{q^N}$ and a \emph{public} set $\mathcal{E}$ of artificial errors $\mathbf{e}$ such that
\begin{equation}\label{eq:ArtError1}
\mathrm{Rk}(\mathbf{e}\mathbf{P}^{-1}\mid \mathbb{F}_{q})\le
t=\left\lfloor\frac{n-k}{2}\right\rfloor,
\end{equation}
where $\mathbf{e}\mathbf{P}^{-1}$ is an error in the intermediate ciphertext \eqref{eq:InterC}.
\paragraph{ Choice of $\mathcal{E}$} The public set of artificial errors is chosen as  the set consisting of all $n$-vectors in $\mathbb{F}_{q^N}^n$ with rank $t_1<t$:
\[
\mathcal{E}=\left\{\mathbf{e}:\mathbf{e}\in \mathbb{F}_{q^N}^n,\mathrm{Rk}(\mathbf{e}\mid \mathbb{F}_q)=t_1\right\}.
\]
\paragraph{ Choice of $\mathbf{P}$} The cryptographer chooses an inverse matrix $\mathbf{P}^{-1}$ in the form $\mathbf{P}^{-1}=\begin{bmatrix} \mathbf{Q}_1 & \mathbf{Q}_2 \end{bmatrix}$, where $\mathbf{Q}_1$ is a submatrix of size $n\times (t-t_1)$ with entries in the \emph{extension field} $\mathbb{F}_{q^N}$ while $\mathbf{Q}_2$ is a submatrix of size $n\times (n-t+t_1)$ with entries in the \emph{base field} $\mathbb{F}_{q}$.
\begin{lemma}
Let $\mathbf{e}$ be any $n$-vector of rank $t_1$. Then the condition Eq. \eqref{eq:ArtError1} is hold.
\end{lemma}
\begin{proof}
We have $\mathbf{e}\mathbf{P}^{-1}=\mathbf{e}\begin{bmatrix} \mathbf{Q}_1 & \mathbf{Q}_2 \end{bmatrix}=\begin{bmatrix} \mathbf{e}\mathbf{Q}_1 & \mathbf{e}\mathbf{Q}_2 \end{bmatrix}$. A vector $\mathbf{e}$ can be represented as $\mathbf{e}=\begin{bmatrix} \mathbf{w}_1 & \mathbf{w}_2 & \dots & \mathbf{w}_{t_1} \end{bmatrix}\mathbf{A}$, where $\mathbf{w}_{j}$'s are linearly independent over $\mathbb{F}_q$ and $\mathbf{A}$ is the $t_1\times n$ matrix over $\mathbb{F}_q$ of rank $t_1$. Then $\mathbf{e}\mathbf{Q}_1 =\begin{bmatrix} \mathbf{w}_1 & \mathbf{w}_2 & \dots & \mathbf{w}_{t_1} \end{bmatrix}\mathbf{B}_1$, where $\mathbf{B}_1=\mathbf{AQ}_1$ is the $t_1\times (t-t_1)$ matrix over the extension field $\mathbb{F}_{q^N}$. It is clear that $\mathrm{Rk}(\mathbf{e}\mathbf{Q}_1\mid \mathbb{F}_q)\le t-t_1$. Similarly, $\mathbf{e}\mathbf{Q}_2 =\begin{bmatrix} \mathbf{w}_1 & \mathbf{w}_2 & \dots & \mathbf{w}_{t_1} \end{bmatrix}\mathbf{B}_2$, where $\mathbf{B}_2=\mathbf{AQ}_2$ is the $t_1\times (n-t+t_1)$ matrix over the \emph{base} field $\mathbb{F}_{q}$. It follows that $\mathrm{Rk}(\mathbf{e}\mathbf{Q}_2\mid \mathbb{F}_q)=\min(t_1,n-t+t_1)\le t_1$. Hence
\[
\begin{array}{l}
\mathrm{Rk}(\mathbf{e}\mathbf{P}^{-1}\mid \mathbb{F}_{q})\le \mathrm{Rk}(\mathbf{e}\mathbf{Q}_1\mid \mathbb{F}_q)+\mathrm{Rk}(\mathbf{e}\mathbf{Q}_2\mid \mathbb{F}_q)\\
\le (t-t_1)+t_1=t=\left\lfloor\frac{n-k}{2}\right\rfloor.
\end{array}
\]
\end{proof}
\begin{remark}
The matrix $\mathbf{P}^{-1}$ can be replaced by a matrix $\mathbf{\widetilde{P}}^{-1}=\mathbf{P}^{-1}\mathbf{Q}$, where $\mathbf{Q}$ is any $n\times n$ non singular matrix over the base field $\mathbb{F}_q)$.
\end{remark}
\begin{example}
Consider again the case when we use a $(28,14)$ rank code with $N=n=28,k=14, q=2, d=15, t=7$.
Possible systems are listed below.

$t_1=0$, $\mathbf{P}$ in the extension field, attacks on PK -- Information sets attacks, brute-force attacks -- not needed, status -- \emph{insecure}.

$t_1=1$, $\mathbf{P}$ in the extension field, attacks on PK -- unknown, brute-force attacks -- $2^{24}$, status -- \emph{insecure}.

$t_1=2$, $\mathbf{P}$ in the extension field, attacks on PK -- unknown, brute-force attacks -- $2^{48}$, status -- \emph{insecure}.

$t_1=3$, $\mathbf{P}$ in the extension field, attacks on PK -- unknown, brute-force attacks -- $2^{72}$, status -- secure.

$t_1=4$, $\mathbf{P}$ in the extension field, attacks on PK -- unknown, brute-force attacks -- $2^{96}$, status -- secure.

$t_1=5$, $\mathbf{P}$ in the extension field, attacks on PK -- unknown, brute-force attacks -- $2^{120}$, status -- secure.

$t_1=6$, $\mathbf{P}$ in the extension field, attacks on PK -- unknown, brute-force attacks -- $2^{144}$, status -- secure.

$t_1=7$, $\mathbf{P}$ in the \emph{base} field, attacks on PK -- Gibson--Overbeck, brute-force attacks -- $2^{168}$, status -- \emph{insecure}.

\end{example}

For $t_1=0\dots 2$, the system is insecure  due to brute-force attacks. For $t_1=7$, the system is insecure because of Gibson--Overbeck's attacks since in this case the matrix $\mathbf{P}$ is in the base field $\mathbb{F}_q$. But for $t_1=3\dots 6$, the system is secure against all known attacks. We recommend the value $t_1=3$, or the value $t_1=4$.

\section{Other variants of the GPT PKC }
We can repeat word for word all previous considerations for variants \eqref{eq:Pub2}- \eqref{eq:Pub4}
and choose for each case a proper column scrambler $\mathbf{P}$ over the \emph{extension field}
$\mathbb{F}_{q^N}$. This prevents Overbeck's and Gibson's attacks.

\section{Conclusion}
An approach is presented to withstand  attacks on the GPT Public key cryptosystem based on rank codes.

It is shown that there exist column scramblers $\mathbf{P}$ over the extension field
$\mathbb{F}_{q^{N}}$ which allow decryption for the authorized party while an unauthorized party can not
break the system by means of known attacks.

\end{document}